\documentclass[a4paper]{amsart}

\usepackage{amsmath}
\usepackage{amssymb}
\usepackage{amsthm}

\usepackage{thmtools} 

\usepackage{ stmaryrd }
\usepackage{mdwtab}
\usepackage{multirow}
\usepackage{mathabx}
\usepackage{synttree}
\usepackage{proof}
\usepackage{graphicx}
\usepackage{tikz-cd}
\usepackage{slashed}
\usepackage{footnote}
\usepackage{pdfpages}

\usepackage{xfrac}

\usepackage{natbib}

\makeatletter
\def\blx@maxline{77}
\makeatother

\usepackage[utf8]{inputenc}

\usepackage[colorlinks,citecolor=brown,linkcolor=blue]{hyperref}

\definecolor{DarkGreen}{RGB}{0,100,0}
\definecolor{DarkBlue}{RGB}{0,0,200}

\newcommand{\ei}[1]{{\textcolor{cyan}{[{#1}---Egor]}}}
\newcommand{\lo}[1]{{\textcolor{red}{[{#1}---Luke]}}}

\newif\ifdraft\drafttrue

\usepackage{microtype}

\usepackage{geometry}
\usepackage{enumitem}

 \usepackage[capitalise]{cleveref}

\newtheorem{theorem}{Theorem}[section]
\newtheorem{proposition}[theorem]{Proposition}
\newtheorem{corollary}[theorem]{Corollary}
\newtheorem{lemma}[theorem]{Lemma}

\newtheorem{fact}[theorem]{Fact}

\theoremstyle{definition}
\newtheorem{definition}[theorem]{Definition}

\newtheorem{example}[theorem]{Example}

\renewcommand{\setminus}{\backslash}

\newcommand{\ria}{\rightarrow}

\newcommand{\up}{\uparrow}

\newcommand{\Up}{\Uparrow}

\newcommand{\CG}{\mathcal{G}}

\newcommand{\CP}{\mathcal{P}}
\newcommand{\CF}{\mathcal{F}}

\newcommand{\sub}{\subseteq}

\newbox\gnBoxA
\newdimen\gnCornerHgt
\setbox\gnBoxA=\hbox{$\ulcorner$}
\global\gnCornerHgt=\ht\gnBoxA
\newdimen\gnArgHgt
\def\Godelnum #1{%
\setbox\gnBoxA=\hbox{$#1$}%
\gnArgHgt=\ht\gnBoxA%
\ifnum     \gnArgHgt<\gnCornerHgt \gnArgHgt=0pt%
\else \advance \gnArgHgt by -\gnCornerHgt%
\fi \raise\gnArgHgt\hbox{$\ulcorner$} \box\gnBoxA %
\raise\gnArgHgt\hbox{$\urcorner$}}



\newcommand\set[1]{{ \{\, #1 \,\} }}

\title{Manipulability of consular election rules}
\author{Egor Ianovski and Mark C. Wilson}

\thanks{Supported by the NZ Marsden fund, grant UOA3706352. The authors acknowledge useful conversations with Arkadii Slinko.}

\newcommand{\best}{\textnormal{best}}
\newcommand{\worst}{\textnormal{worst}}

\newcommand{\Pup}{P^{i\up s}}

\newcommand{\Pip}{P_i^{\up s}}
\newcommand{\PIp}{P_i^{\Up s}}
\newcommand{\PUip}{P^{i\Up s}}
\newcommand{\PUp}{P^{\Up s}}
\newcommand{\Pnp}[1]{P_i^{#1\up s}}

\begin{document}

\maketitle

\begin{abstract}
The Gibbard-Satterthwaite theorem is a cornerstone of social choice theory, stating
that an onto social choice function cannot be both strategy-proof and non-dictatorial
if the number of alternatives is at least three. The Duggan-Schwartz theorem
proves an analogue in the case of set-valued elections: if the function is onto
with respect to singletons, and can be manipulated by neither an optimist nor
a pessimist, it must have a weak dictator. However, the assumption that the function
is onto with respect to singletons makes the Duggan-Schwartz theorem inapplicable
to elections which necessarily select a committee with multiple members. In this
paper we make a start on this problem by considering elections which elect a committee of size two
(such as the consulship of ancient Rome). We establish that if such a \emph{consular
election rule} cannot be expressed as the union of two disjoint social choice
functions, then strategy-proofness implies the existence of a dictator. Although we suspect that a similar result holds for larger sized committees, there appear to be many obstacles to 
proving it, which we discuss in detail.
\end{abstract}

\section{Introduction}

\begin{flushright}{\slshape    
    From that time on Caesar managed all the affairs of state alone and after his own pleasure; so that sundry witty fellows, pretending by way of jest to sign and seal testamentary documents, wrote ``Done in the consulship of Julius and Caesar.''}\\ \medskip
    ---  Suetonius, Lives of the Caesars. Chapter XX.
\end{flushright}

\citet{DuFa1961} conjectured what later became the Gibbard-Satterthwaite theorem \citep{Gibbard1973,Satterthwaite1975}, one form of which is: if a social choice function has at least 3 viable alternatives and is strategy-proof, then it is dictatorial. This result, though fundamental, is not applicable in many common situations, because symmetric voting rules must have ties. \citet{Gibb1977} extended the result to the case where ties are broken randomly, showing that such a rule is strategy-proof only if it is a mixture of rules that can only be affected by one voter, and rules that
limit the outcome to two alternatives. While Gibbard
did not mention it as such, this paper
represents one of the first attempts at studying
strategy-proof social choice \emph{correspondences},
which output a set of alternatives, as opposed to
a social choice \emph{function} which outputs only one.

The fundamental difficulty in studying the manipulability of social choice correspondences is that, in
a certain sense, the problem is ill-defined. To manipulate a social choice function, a voter submits
an insincere preference ordering to obtain an outcome
that is better than he would have obtained had he voted
sincerely. To verify that the outcome is better, we
can compare it against the voter's sincere ballot.
In the case of a social choice correspondence the voter submits a preference order
over alternatives, but the function produces a set
of alternatives. We are not given enough information
to deduce how the voter will assess the outcome, as
there is no unique way to extend an order over a set
to an order over the power set. Indeed, the subject of \emph{set ranking} has been studied 
in detail (see \cite{BBP2004} for a survey).

A common approach taken by the early papers on strategy-proof social choice correspondences (e.g. \cite{Barbera1977}; \cite{Kelly1977}; \cite{Pattanaik1978}) is to assume some properties on what an extension from an ordering of alternatives to an ordering of sets must satisfy, and show that it implies some notion of dictatoriality, although
to achieve these results the authors had to make further assumptions about the domains and behaviour of their functions. For example, under the \emph{positive response} assumption of \citet{Barbera1977} if a set $X$ is elected at a profile $P$, and $P'$ is obtained by a single voter bumping some $x\in X$ up their ballot, $x$ becomes the unique winner under $P'$; this condition is satisfied by, for example, Borda's rule, but is not satisfied by many other natural rules. The Duggan-Schwartz theorem \citep{Duggan1992} stood out from these results because of how little it had to assume -- that the correspondence be onto with respect to singletons -- to achieve their result. That said, the Duggan-Schwartz theorem is difficult to compare with the previous work, because there is a subtle
difference in the frameworks used. Duggan and Schwartz do not consider an extension of a preference order,
but of a notion of manipulability. The two are not
interchangeable as there are certain properties we expect
from an order -- at the very least, transitivity,
desirably, completeness, and ideally, anti-symmetry -- whereas just about any binary relation can serve
for manipulability. In particular, the optimistic/pessimistic manipulation of Duggan-Schwartz
does not induce a transitive preference order, as it is possible for a voter
to optimistically prefer $X$ to $Y$ and pessimistically prefer $Y$ to $X$ (see Section~\ref{s:prelim} and \ref{s:reduce} for more details). For a more thorough overview of the history of strategy-proofness, the reader is recommended to see
\citet{Barbera2010}.

In this paper we concentrate on social choice correspondences that must always output a set of fixed size $k$, which we call \emph{committee selection rules} \citep{Skowron2016}. In this setup, the assumptions of the abovementioned papers are unnatural or vacuous -- for example, being onto with respect to singletons can never happen. Thus a different approach is needed. This is not the first paper on strategy-proof
committee selection rules. \citet{Ozyurt2008} and \citet{Reffgen2011} are the most relevant to the enquiry pursued here. The general
approach of those two papers is the same: demonstrate
that a certain extension of a preference order from $A$ to fixed-size committees over $A$ results in what
\citet{Aswal2003} define as a \emph{linked} domain, and
appeal to a theorem of \citet{Aswal2003} that states
that a unanimous social choice function on a linked domain is dictatorial. \citet{Ozyurt2008} demonstrate
that all \emph{reasonable} extensions, roughly speaking
extensions which have a clear first and second choice,
result in linked domains, and \citet{Reffgen2011} deals with extensions that satisfy the axiom used by 
\citet{Kelly1977}. Neither paper implies our result, but
our general approach is the same -- construct a linked domain and invoke \citet{Aswal2003}.

To obtain our result we use a lexicographic extension of preferences over
alternatives to preferences over sets: a voter compares $X$ and $Y$ by
first comparing his favourite elements in $X$ and $Y$, then the second favourite,
and so on. While this is not equivalent to the Duggan-Schwartz conditions, in the
case of a rule that elects a pair of alternatives it is strictly narrower -- every
lexicographic manipulation is either an optimistic or a pessimistic manipulation. Thus an impossibility result for lexicographic manipulations implies, \emph{a fortiori}, an impossibility result for pessimistic/optimistic manipulation. \citet{Ozyurt2009} and \citet{Campbell2002} also work in the framework of lexicographic manipulation and their results imply a weakened version of our main theorem -- an impossibility result for manipulating an
\emph{onto} selection rule, but do not imply our more general result.

\subsection{Our contribution}
\label{ss:contrib}

We work in the manipulation framework of Duggan-Schwartz: set-valued functions that cannot be manipulated by an optimist and cannot be manipulated by a pessimist. In particular, we are interested in
committee selection rules that select a committee of size two. Since the Duggan-Schwartz assumption that
a correspondence is onto with respect to singletons
is meaningless here, we replace it with the assumption that it is possible for every alternative
to be on some winning committee. This does not imply
that all possible committees lie in the range of the function, as we demonstrate in \cref{fact:4notonto}. This, in turn, forces us to consider a new notion of dictatorship (\cref{def:rangedictator}). Roughly speaking, a dictator is a voter who gets his ``best'' outcome no matter what, but the difficulty is in showing that this notion is well-defined.

We start with the observation that such rules can be defined in two categories: those that can be
reduced to the union of two disjoint social choice
functions (\emph{reducible} rules) and those that
cannot (\emph{irreducible}). It turns out that irreducible rules can have a non-trivial structure
to their range, which we characterise by modelling
the range as a graph. This interpretation may be useful in the future, as it establishes an interesting connection between strategy-proofness and a structural property of graphs (\cref{fact:edgeconnected}).

Our main result (Corollary~\ref{cor:main}) shows that  for irreducible rules, strategy-proofness for optimists and pessimists implies that some voter always gets his favourite viable committee elected.

For reducible rules, we observe that the
Gibbard-Satterthwaite theorem dictates that if the number of viable alternatives is at least five
then the function will have a ``partial'' dictator on part of the domain, but it may not have an overall
dictator -- these rules are non-dictatorial, but in
a very trivial sense. In the special case where one of the component social choice functions has a domain of size
one, i.e. it always elects the same alternative, some voter always gets his favourite committee
elected if the number of viable alternatives is at least four.

\section{Preliminaries}
\label{s:prelim}

We use $S_k(X)$ to denote the set of all $k$-element subsets of a set $X$. When we are talking about a preference order $P_i$, we use the standard infix notation: $\succeq_i$ for the non-strict order, $\succ_i$ for the strict counterpart, and
$x\sim_i y$ for indifference.

\begin{definition}
    Let $V$ be a finite set of voters, $A$ a finite set of alternatives.
    
    A profile $P$ consists of a linear order over $A$ (also known
    as a \emph{preference order} or a \emph{ballot}), $P_i$, for every voter $i$. The set of all profiles of voters $V$ over alternatives $A$ is denoted $\CP(V,A)$. We use $P_{-i}$ to refer to the ballots of all voters except $i$. Hence,
    $P=P_iP_{-i}$ and $P_i'P_{-i}$ is obtained from profile $P$ by replacing
    $P_i$ with $P_i'$.
    
    A \emph{social choice function} maps a profile to a single alternative,
    $F:\CP(V,A)\ria A$. A \emph{social choice correspondence} produces a nonempty set
    of alternatives, $F:\CP(V,A)\ria 2^A\setminus\set{\emptyset}$. A \emph{$k$-committee selection rule} produces a set of alternatives of size $k$, $F:\CP(V,A)\ria S_k(A)$.
    A $2$-committee selection rule is called a \emph{consular election rule}.
\end{definition}

\begin{definition}\label{def:sp}
    Let $\emptyset\neq W\subseteq A$. We use $\text{best}(P_i,W)$ to denote the best alternative in $W$ according to $P_i$, $\text{worst}(P_i,W)$ the worst.
    
    We extend $\succeq_i$ into three relations over $2^A\setminus\set{\emptyset}$:
    \begin{enumerate}
        \item 
$X\succeq_i^O Y$ iff $\best(P_i,X)\succeq_i\best(P_i,Y)$.
        \item 
$X\succeq_i^P Y$ iff $\worst(P_i,X)\succeq_i\worst(P_i,Y)$.
        \item 
$X\succ_i^{DS} Y$ iff $X\succ_i^O Y$ or $X\succ_i^P Y$.
    \end{enumerate}
    
    A social choice function is \emph{strategy-proof}  if for all $P_i'$, whenever $F(P_iP_{-i})=a$ and $F(P_i'P_{-i})=b'$,
    $a\succeq_i b$.   
    A committee selection rule is \emph{strategy-proof for optimists} (SPO) if for all $P_i'$, whenever $F(P_iP_{-i})=W$ and $F(P_i'P_{-i})=W'$,
    $W\succeq_i^O W'$.   
    A committee selection rule is \emph{strategy-proof for pessimists} (SPP) if for all $P_i'$, whenever $F(P_iP_{-i})=W$ and $F(P_i'P_{-i})=W'$,
    $W\succeq_i^P W'$.
\end{definition}

Duggan and Schwartz define viability to be the notion that for every $a\in A$, there
exists a $P$ such that $F(P)=\set{a}$. This is reasonable in the context of a social choice correspondence -- an election rule that allows a certain candidate to be tied for victory, but never to win outright, seems rather strange. In the case of $k$-committee selection rules it is a vacuous concept.
We use a weaker notion of viability where every alternative has the chance to appear on some winning committee, but we make no assumptions on whom the alternative might have to share the privilege with.

\begin{definition}
    A committee selection rule satisfies \emph{weak viability} if for every $a\in A$, there exists $P$ such that $a\in F(P)$.
    
    A committee rule is \emph{unanimous} if whenever $\best(P_i,A)=a$ for all $i$, necessarily $a\in F(P)$ (if all the voters agree on the best alternative, it is on the committee).
    
    A committee rule satisfies \emph{veto} if whenever $\worst(P_i,A)=a$ for all $i$, necessarily $a\notin F(P)$ (if the voters agree on the worst alternative, it is not on the committee).
\end{definition}

The notion of a dictator in the case of a social choice function is straightforward -- every voter has a unique first choice alternative, so if that voter is a dictator it stands to reason that the first choice will be elected. In the literature on social choice correspondences this notion is typically extended by defining a dictator to be a voter who always gets an outcome that is maximal in the extension of his preferences to sets.

However, Duggan-Schwartz manipulability does not extend to a preference order over $2^A\setminus\set{\emptyset}$.
To see this, consider a voter with preferences $a\succ_i b\succ_i c\succ_i d\succ_i e\succ_i f$, and observe that $\set{a,e}\prec_i^{DS}\set{c,d}$ because $\set{a,e}\prec_i^P \set{c,d}$, and $\set{c,d}\prec_i^{DS}\set{b,f}$ because $\set{c,d}\prec_i^O{b,f}$, but it is not the case that $\set{a,e}\prec_i^{DS}\set{b,f}$.

There are, however, two natural notions of dictatorship we can define without appealing to preferences over sets.

\begin{definition}
    A social choice function $F$ is \emph{dictatorial} if there is some $i\in V$ such that $F(P)$ is
    always the first choice of $i$.
    Given a committee selection rule $F$, a \emph{weak dictator} is some $i\in V$ such that the first choice of $i$ is always in $F(P)$.    
    Given a $k$-committee selection rule $F$, a \emph{strong dictator} is some $i\in V$ such that $F(P)$ consists of the top $k$ choices of $i$.   
\end{definition}

Note that there is a duality at play here.  A weak dictator is precisely a voter who gets a maximal outcome under $\succeq_i^O$ and a strong dictator gets a maximal (in fact, the maximum) outcome under $\succeq_i^P$. This suggests a natural definition of a dictator in the Duggan-Schwartz case: a voter that gets an outcome that is maximal under both preference orders. However, as at this point we have no guarantee that such a notion is well-defined (a priori, there is no reason to assume that there is an element in the range of an election that is maximal under both orders), we shall postpone introducing it until \cref{def:rangedictator}.

\section{Monotonicity}

It is common to obtain impossibility results
in social choice by showing that the concept in question implies some sort of monotonicity on ballots.
This typically takes the form that swapping two alternatives in some ballot either does
not change the outcome of an election, or changes it in a very predictable way.
Our approach is no different. We will show that in elections which satisfy SPP and SPO,
moving an alternative up or down a ballot affects only whether that alternative
features on the winning committee or not.

\subsection{Auxiliary results}

\begin{definition}
    $P^{i\up s}$ is the profile obtained from $P$ by swapping $s$ in $P_i$ with the alternative directly above it, and  $P_i^{\up s}:=(P^{i\up s})_i$ is the resulting order for voter $i$.
\end{definition}

We start with a simple observation: if we swap a candidate upward, it may become the most preferred but will not change the least preferred element. Note that $\best(\Pip,W)=s$ can only occur if $s$ is the second-most preferred element of $W$.

\begin{lemma}\label{lem:bounds}
    Let $W\sub A$. Let $t:=\best(P_i,W)=t$, $b:=\worst(P_i,W)$, and suppose $s\neq t, b$. Then $\best(\Pip,W)=t$ or $s$, and $\worst(\Pip,W)=b$.   
    Likewise, let $t:=\best(\Pip,W)$, $b:=\worst(\Pip,W)=b, s\neq t, b$. Then $\best(P_i,W)=t$ and $\worst(P_i,W)=b$ or $s$.
\end{lemma}

If there are at least two elements in the chosen committee and $F$ is strategy-proof, then swapping a third element upward in one preference order leads to a small list of possibilities.

\begin{lemma}\label{lem:sout}
    Let $F$ be a committee selection rule satisfying SPO and SPP, $F(P)=W$, $\best(P_i,W)=t$, $\worst(P_i,W)=b$, $t\neq b$.
    
    Let $s\neq t,b$, $F(P^{i\up s})=W'$. The following are true:
    \begin{enumerate}
        \item
        $\best(\Pip,W')=t$ or $s$.
        \item
        $\worst(\Pip,W')=b$ or $s$.
    \end{enumerate}
\end{lemma}
\begin{proof}
    Let $a:=\best(\Pip,W')$ and suppose $a\neq t, a\neq s$. Suppose $a\succ_i t$. By \cref{lem:bounds}, $\best(P_i,W')= a$, and hence voter $i$ has an optimistic manipulation from $P$ to $\Pup$.
Now suppose $a\prec_i t$. By \cref{lem:bounds}, $\best(\Pip,W)=s$ or $t$. If it is $t$, voter $i$ has an optimistic manipulation from $\Pup$ to $P$. If it is $s$ then, seeing how $t\in W$, it must be the case that $s\succ_i^{\up s} t\succ_i^{\up s} a$, and we have an optimistic manipulation again. This yields (i).
    
    For (ii), let $a:=\worst(\Pip,W')$, $a\neq b,s$. Suppose $a\succ_i b$. By \cref{lem:bounds}, $\worst(P_i,W')=a$ or $s$. If it is $a$, voter $i$ has a pessimistic manipulation from $P$ to $\Pup$. If it is $s$, then we consider two further cases. If $s\succ_i b$, we have a pessimistic manipulation from $P$ to $\Pup$. If $s\prec_i b$, observe the following:
    \begin{enumerate}
        \item
        $s\in W'$, else it could not be the worst element of $W'$ under $P_i$.
        \item
        $s\prec_i b\prec_i a$. In other words, $s$ is at least two positions below $a$ in $P_i$.
        \item
        $\worst(\Pip,W')=a$.
    \end{enumerate}
    This gives us a contradiction -- $s$ moved up one position from $P_i$ to $\Pip$ so, given 2, $s\prec_i^{\up s} a$. Given 1, that $s\in W'$, this contradicts 3. In sum, if $s\prec_i b$ then $\worst(\Pip,W')\neq a$.
    
    Suppose $a\prec_i b$. By \cref{lem:bounds}, $\worst(\Pip,W)=b$. Voter $i$ has a pessimistic manipulation from $\Pup$ to $P$.
\end{proof}

Given a strategy-proof rule, swapping a voter's most preferred element of the 
chosen committee member upward in that voter's order preserves that voter's most and least preferred members of the committee.

\begin{lemma}\label{lem:stop}
    Let $F$ be a committee selection rule satisfying SPO and SPP, $F(P)=W$, $\best(P_i,W)=s$, $\worst(P_i,W)=b$.
    
    Let $F(P^{i\up s})=W'$. The following are true:
    \begin{enumerate}
        \item
        $\best(\Pip,W')=s$.
        \item
        $\worst(\Pip,W')=b$.
    \end{enumerate}
    \end{lemma}
    \begin{proof}
    Let $\best(\Pip,W')=a$, $a\neq s$. Suppose $a\succ_i s$. By \cref{lem:bounds}, $\best(P_i,W')=a$. Voter $i$ has an optimistic manipulation from $P$ to $\Pup$. Suppose $a\prec_i s$. By \cref{lem:bounds}, $\best(\Pip,W)=s$ or $s$, that is $\best(\Pip,W)=s$ . Voter $i$ has an optimistic manipulation from $\Pup$ to $P$.
    
    Let $\worst(\Pip,W')=a$, $a\neq b$. Suppose $a\succ_i b$. By \cref{lem:bounds}, $\worst(P_i,W')=a$ or $s$. If it is $s$, observe that if $s\neq b$, $s\succ_i b$, giving voter $i$ a pessimistic manipulation from $P_i$ to $\Pup$. It follows that in this case $s=b$, which satisfies the lemma statement. If it is $a$, voter $i$ has a pessimistic manipulation from $P$ to $\Pup$.
    
    Suppose $a\prec_i b$. By \cref{lem:bounds}, $\worst(\Pip,W)=b$. Voter $i$ has a pessimistic manipulation from $\Pup$ to $P$.
\end{proof}

Note that the preceding lemmata applied to all committee selection rules. The following applies
only to consular election rules. It says that for a strategy-proof rule, swapping a voter's worst element of the elected committee upward in that voter's order does not change the elected committee.

\begin{lemma}\label{lem:sbot2}
    Let $F$ be a consular election rule satisfying SPO and SPP, $F(P)=W$, $\best(P_i,W)=t$, $\worst(P_i,W)=s$.
    
    Let $F(P^{i\up s})=W'$. It follows that:
    \begin{enumerate}
        \item
        $\best(\Pip,W')=t$ or $s$.
        \item
        $\worst(\Pip,W')=s$ or $t$.
    \end{enumerate}
\end{lemma}
\begin{proof}
    Let $\best(\Pip,W')=a$, $a\neq t,s$. Suppose $a\succ_i t$. By \cref{lem:bounds}, $\best(P_i,W')=a$. Voter $i$ has an optimistic manipulation from $P$ to $\Pup$.
    
    Suppose $a\prec_i t$. By \cref{lem:bounds}, $\best(\Pip,W)=t$ or $s$.  If it is $t$, voter $i$ has an optimistic manipulation from $\Pup$ to $P$. If it is $s$, we consider two cases. If it is the case that $s\succ_i^{\up s} a$, voter $i$ has an optimistic manipulation from $\Pup$ to $P$. If it is the case that $s\prec_i^{\up s} a$, note that the assumption that $\best(\Pip,W)=s$ implies that $s\succ_i^{\up s} t$. This is a contradiction as $a\prec_i^{\up s}t$.
    
    Let $\worst(\Pip,W')=a$, $a\neq s, t$. Suppose $a\succ_i s$. By \cref{lem:bounds}, $\worst(P_i,W')=a$ or $s$. If it is $a$, voter $i$ has an optimistic manipulation from $P$ to $\Pup$. If it is $s$, then as $|W'|=2$ it follows that $W'=\set{s,a}$. As $\worst(\Pip,W')=a$, $\best(\Pip,W')=s$. Now observe that $a\succ_i s$, $s\succ_i^{\up s}a$, implies that $t$ is at least two positions above $s$ in $P_i$. This means that $t\succ_i^{\up s} s$, and voter $i$ has an optimistic manipulation from $\Pup$ to $P$.
    
    Suppose $a\prec_i s$. By \cref{lem:bounds}, $\worst(\Pip,W)=s$, which gives voter $i$ a pessimistic manipulation from $\Pup$ to $P$.
\end{proof}

\subsection{Key monotonicity lemmata}

We now show that for a strategy-proof consular election rule, swapping an alternative $s$ upward in some voter's order usually does not change the elected committee, and the only possible change is that $s$ overtakes some committee member and replaces it.

\begin{lemma}[Upwards monotonicity]\label{lem:upmono}
    Let $F$ be a consular election rule satisfying SPP and SPO.
    
    Let $F(P)=\set{a,b}$. Let $P_i'$ be the ballot obtained from $P_i$ by moving $s$ up some number of positions.
    
    If $s=a$ or $s=b$:
    $$F(P_i'P_{-i})=F(P).$$
    
    If $s\neq a,b$, either:
    \begin{enumerate}
    \item
    $F(P_i'P_{-i})=F(P).$
    \item
    $F(P_i'P_{-i})=\set{s,b}.$
    \end{enumerate}
    Furthermore, 2 happens only if $s$ overtakes $a$ from $P_i$ to $P_i'$.
\end{lemma}
\begin{proof}
    Let $s=a$ or $b$. If $\best(P_i,\set{a,b})=s$, then by applying \cref{lem:stop}
    repeatedly, we see that no matter how many times we bump $s$ up one position,
    the winning committee remains unchanged. If $\worst(P_i,\set{a,b})=s$, then
    we can apply \cref{lem:sbot2} repeatedly to show that as long as $s$ remains
    the least favourite element in the set, the winning committee remains unchanged.
    Once $s$ becomes the favourite element in the set, we can apply \cref{lem:stop}.
    
    Now suppose $s$ is distinct from $a$ and $b$. By applying \cref{lem:sout}, we
    can see that each time we move $s$ up one position, either the two
    elements in the set remain unchanged, or one of them is replaced by $s$. In
    the former case, we get 1 from the lemma statement. In the latter, we get 2.
    It remains to show that the latter case can only happen if $s$ overtakes
    the element it supplants from $P_i$ to $P_i'$.
    
    Suppose, for contradiction, that this is not the case. That is, $F(P)=\set{a,b}$,
    $F(P_i'P_{-i})=\set{s,b}$ and either $a\succ_i s$ and $a\succ_i' s$, or $s\succ_i a$ and $s\succ_i' a$.
    
    First let $a\succ_i s$ and $a\succ_i' s$. Suppose $\worst(P_i',\set{a,b})=a$. By
    transitivity, $\worst(P_i',\set{s,b})=s$. As $a\succ_i' s$, this gives $i$
    a pessimistic manipulation from $P_i'$ to $P_i$. Now suppose $\best(P_i',\set{a,b})=a$.
    It follows that $a\succ_i'\best(P_i',\set{s,b})$. This gives $i$ an optimistic manipulation
    from $P_i'$ to $P_i$.
    
    Now, let $s\succ_i a$ and $s\succ_i' a$. Suppose $\best(P_i,\set{a,b})=a$. By
    transitivity, $\best(P_i,\set{s,b})=s$. As $a\succ_i s$, this gives $i$
    an optimistic manipulation from $P_i$ to $P_i'$. Now suppose $\worst(P_i,\set{a,b})=a$.
    It follows that $\worst(P_i,\set{s,b})\succ_i a$. This gives $i$ a pessimistic manipulation
    from $P_i$ to $P_i'$.
    
    As such, SPP and SPO prevents the possibility of $s$ entering the winning committee without
    overtaking the element it replaces. As such, 2 can only happen if the overtaking takes place.
\end{proof}
\begin{corollary}
    Let $F$ be a weakly viable consular election rule satisfying SPP and SPO.
   Then $F$ is unanimous.
\end{corollary}
\begin{proof}
Let $a\in A$. By weak viability there is some profile $P$ with $a\in F(P$. By Lemma~\ref{lem:upmono}, if we swap $a$ to the top of all voters' preference orders, the elected committee does not change.    
\end{proof}

Similarly, if we move $s$ downward in some voter's order, $s$ may be replaced on the committee by an alternative that overtook it, and otherwise there is no change to the committee.

\begin{lemma}[Downwards monotonicity]\label{lem:downmono}
    Let $F$ be a consular election rule satisfying SPP and SPO.

    Let $F(P)=\set{a,b}$. Let $P_i'$ be the ballot obtained from $P_i$ by moving $s$ down some number of positions.
    
    If $s\neq a,b$:
    $$F(P_i'P_{-i})=F(P).$$
    
    If $s=a$ or $s=b$ (WLOG, $a$), either:
    \begin{enumerate}
    \item
    $F(P_i'P_{-i})=F(P).$
    \item
    $F(P_i'P_{-i})=\set{c,b}, c\neq s.$
    \end{enumerate}
    Furthermore, 2 happens only if $s$ drops below $c$ from $P_i$ to $P_i'$.
\end{lemma}
\begin{proof}
    If $s\neq a,b$, then observe that if $s$ moves down the ballot from $P_i$ to
    $P_i'$, then it moves up the ballot from $P_i'$ to $P_i$. By upwards monotonicity,
    either $F(P)=F(P_i'P_{-i})$ or $s\in F(P)$. As $s\notin\set{a,b}$, it follows
    that $F(P)=F(P_i'P_{-i})$.
    
    If $s=a$, we consider four possibilities for $F(P_iP_{-i})$:
    \begin{itemize}
        \item 
            $F(P_iP_{-i})=\set{s,b}$ (committee unchanged).
        \item
            $F(P_iP_{-i})=\set{b,c}$ ($b\neq s$ is unchanged).
        \item
            $F(P_iP_{-i})=\set{s,c}$ ($s$ is unchanged).
        \item
            $F(P_iP_{-i})=\set{c,d}$ (both elements on committee changed).
    \end{itemize}
    In case 1, we start at $P_iP_{-i}$ and move $s$ up the ballot. By upwards
    monotonicity, we find that $F(P)=F(P_iP_{-i})$.
    
    In case 2, we start at $P_iP_{-i}$ and move $s$ up the ballot. By upwards
    monotonicity, if $s\in F(P)$ then it can only enter after it overtakes
    the element it replaces. As $F(P)=\set{s,b}$, it follows that $s$
    enters after it overtakes $c$, and as a consequence $F(P_iP_{-i})=\set{b,c}$
    only if $s$ drops below $c$ as we move from $P_i$ to $P_i'$.
    
    In case 3, by upwards monotonicity $F(P)=\set{s,c}$. As such, case 3
    cannot occur.
    
    In case 4, by upwards monotonicity $F(P)$ is one of $\set{c,d}$, $\set{s,c}$
    or $\set{s,d}$. As it is none of these, case 4 cannot occur.
\end{proof}

\section{Strategy-Proofness and Reducible elections}
\label{s:reduce}

\subsection{Two classic impossibilities}

We first state the classic result in the area, for committees of size $1$.

\begin{theorem}[\citet{Gibbard1973}; \citet{Satterthwaite1975}]
    Let $F$ be an onto, strategy-proof social choice function with $|A|\geq 3$. Then $F$ is dictatorial.
  
\end{theorem}

Duggan and Schwarz proved an analogue of this for variable sized committees. In fact they give two rather different-looking versions in  \citet{Duggan1992, Duggan2000}. \citet{Taylor2002} gave a more readable presentation which has been widely described as the ``Duggan-Schwartz Theorem", and we use this version.

\begin{theorem}
    Let $F$ be a committee selection rule that satisfies SPP, SPO and is onto with respect
    to singletons. That is, for every $a\in A$ there exists a $P$ such that $F(P)=\set{a}$.    
    Then for $|A|\geq 3$, $F$ has a weak dictator.
 
\end{theorem}

We aim to extend this last result to the case of consular election rules, replacing being onto for singletons with weak viability. This work will not be as straightforward as expected, and will take until Section~\ref{s:linked} to complete. We first deal with some annoying trivialities.

\subsection{The case of Marius}
\label{ss:marius}

One might, in light of the Duggan–Schwartz theorem, assume that the Gibbard-Satterthwaite
theorem can be extended directly extended to committee selection rules:
with $|A|\geq 3$, a weakly viable committee selection rule satisfying SPP and SPO
has a dictator of some sort. However, it is not difficult to construct a counterexample.

\begin{definition}
    A committee selection rule is \emph{Marian} if there exists an $a\in A$ such that $a\in F(P)$ for all $P$. Such an $a$ is called a $Marius$.\footnote{Gaius Marius held the consulship of Rome an unprecedented seven times. If you are electing a pair of consuls and one of the candidates is Marius, you are in fact electing a single consul.}
\end{definition}

\begin{fact}
    For $|A|\geq 3$, there exist Marian consular election rules satisfying SPO, SPP and weak viability without a weak dictator.
\end{fact}
\begin{proof}
    Majority vote on the second consul with lexicographic tie-breaking.
\end{proof}

Of course, what is going on here is that even though three alternatives are present,
only two sets are in the range of the function. If we increase the number of alternatives to four, and hence the size of the range to three, we recover a form of dictatoriality.

\begin{proposition}\label{prop:Mariandictator}
    If $|A|\geq 4$, every Marian consular election rule that satisfies SPO, SPP and weak viability has a weak dictator.
\end{proposition}
\begin{proof}
    Let $m$ be Marius. Consider the social choice function $F':\CP(V,A\setminus\set{m})\ria A\setminus\set{m}$ such that $F'(P)=a$ where $a\in F(P^*), a\neq m$. $P^*$ is obtained from $P$ by ranking $m$ last in every ballot. Once we establish the following we can invoke the Gibbard-Sattherwaite theorem:
    \begin{enumerate}
        \item
        $|A\setminus\set{m}|\geq 3$.
        \item
        $F'$ is onto.
        \item
        $F'$ is strategy-proof.
    \end{enumerate}
    1 follows from $|A|\geq 4$ and 2 follows from weak viability.
    
    For 3, assume there is a $P\in\CP(N,A\setminus\set{m})$ such that there is some $P_i$ for which $F'(P)=a\prec_i F'(P_iP_{-i})=b$. Extend $P$ to $P^*$ by ranking $m$ last in every ballot. It follows that $F(P^*)=\set{a,m}$ and $ F((P_iP_{-i})^*)=\set{b,m}$, giving voter $i$ an optimistic manipulation.
\end{proof}
The reader will note that the above proof only used the SPP assumption. It could likewise be rewritten to only use SPO. This should not be surprising as Marian election rules are, essentially, social choice functions in disguise, and manipulating
a singleton by an optimist or a pessimist is the same thing.

We have earlier shown that strategy-proof consular election rules are unanimous. With
the further assumption that they are non-Marian, we can show that they satisfy veto.

\begin{fact}
    Let $F$ be a non-Marian, weakly viable consular election rule satisfying SPP and SPO.
    Then $F$ satisfies veto.
\end{fact}
\begin{proof}
    Let $P$ be any profile where $a$ loses. Such a profile must exist because
    $F$ is non-Marian. Once $a$ is moved to the bottom of every ballot it still
    loses by monotonicity.
\end{proof}

\subsection{The general case}

Marian election rules are a special case of a more general class of elections.

\begin{definition}
    A $k$-committee selection rule $F$ is said to be \emph{reducible} if there exists an integer $j$ and a partition of $A$ into two sets, $A=B\uplus C$, such that $F(P)=G(P|B)\cup H(P|C)$ for a $j$-committee selection rule $G:\CP(V,B)\ria S_j(B)$ and a $(k-j)$-committee selection rule $H:\CP(V,C)\ria S_{k-j}(C)$.
\end{definition}

Note that in many cases this is an entirely natural way to select a committee --
we might want a committee formed of a fixed number of men and women, plebeians and patricians,
Lords Temporal and Lords Spiritual. However, from a formal point of view these are not very interesting; a reducible $k$-committee selection rule is not, in some sense, a ``true'' $k$-committee selection rule, and it does not add anything to the theory. Reducible rules also complicate any attempt to gracefully extend the
Duggan-Schwartz theorem.

\begin{fact}
    There exist non-Marian consular election rules with $|A|=4$ that satisfy SPP, SPO and weak viability and do not have a weak dictator.
\end{fact}
\begin{proof}
    Divide $A$ into two halves and have majority voting on each.
\end{proof}

It is true that for $|A|\geq 5$ every reducible, strategy-proof and viable consular election rule 
has a 
``partial'' dictator, as the Gibbard-Satterthwaite theorem ensures that one 
of the two 
social choice functions must be dictatorial. However, it will not have a weak dictator over 
the entire domain unless the same voter is the dictator of both social choice functions. 
Thus, in a very boring sense, the Duggan-Schwartz theorem fails. We can have
committee selection rules that are both strategy-proof and non-dictatorial.

It is worth noting that a reducible consular election rule is strategy-proof
in the Duggan-Schwartz sense precisely when the component functions are strategy-proof
in the classical sense.

\begin{proposition}\label{fact:reduciblesp}
    Let $F$ be a consular election rule that is reducible
    to social choice functions $G$ and $H$. Then $F$ satisfies SPP and SPO if and
    only if $G$ and $H$ are strategy-proof.
\end{proposition}
\begin{proof}
    Let $B\uplus C= A$, $F(P)=G(P|B)\cup H(P|C)$.
    
    Suppose $F$ satisfies SPP and SPO, and let voter $i$ have a profitable deviation from $Q_i$ to $Q_i'$ in $G$. That is,
    $G(Q_iQ_{-i})=a$, $G(Q_i'Q_{-i})=b$ and $b\succ_i a$. Let $P$ be any profile
    whatsoever over $A$ for which $P|B=Q$. Observe that $F(P)=\set{a,x}$, for
    some $x\in C$. Modify $P$ into $R$ by moving all the alternatives in $B$
    to the top of the ballots, without changing the order within $B$ and $C$.
    Applying \cref{lem:upmono} we see that $F(R)=F(P)$ -- $a$ remains a winning
    alternative because no alternative overtakes it, so the only way the winning
    committee can change is if some $y\in B$ replaces $x$, but $F(R)=\set{a,y}$
    is impossible because one element of $F(R)$ must come from $B$ and one from
    $C$, and $a$ already comes from $B$.
    
    Now consider what happens when voter $i$ deviates from $R_i$ to $R_i'$
    with the property that $R_i'|B=Q'$ and $R_i'|C=R_i|C$. Since $R_i'|B=Q'$,
    $G$ must elect $b$, and $R_i'|C=R_i|C$ implies $H$ still elects $x$.
    Thus $F(R_i'R_{-i})=\set{b,x}$. If $b\succ_i x$, this is an optimistic
    manipulation, and if $x\succ_i b$, this is a pessimistic manipulation, neither
    of which can take place.
    
    For the other direction, suppose $G$ and $H$ are strategy-proof. The only
    possible functions these could be are dictatorship, imposed rule, and
    majority voting between two alternatives. It is easy to see that a
    disjoint union of any of these cannot be manipulated by either an optimist
    or a pessimist.
\end{proof}

\section{The range of an election}\label{sec:graphs}

Our quest for an impossibility result has had a rocky start.
Marian elections only obtain a weak dictator with at least four alternatives, and reducible elections in general need at least
five alternatives for an even weaker notion of a partial dictator.

In response to this, we can argue that so far we have only considered
a somewhat degenerate class of committee selection rules -- sure, we
can obtain a consular election rule by gluing two social choice functions
together, and in some cases it may even be natural to do so, but the
behaviour of such an election is determined entirely by the behaviour
of its components, and studying it sheds no light on the mathematics
of selecting a committee in general. What we would like to do
is to understand elections that are not reducible.

Intuitively, reducible rules occur because their range does not cover sufficiently
many possible committees. It is clear that, for $|A|\geq 3$, an onto rule is irreducible.
It would have been nice had this been the end of the story -- perhaps the combination of
strategy-proofness, viability and irreducibility are sufficiently strong conditions to force
every committee to be viable?

We start with a counterexample.

\begin{fact}\label{fact:4notonto}
    There exist irreducible consular election rules satisfying SPP, SPO and weak viability that are not onto.
\end{fact}
\begin{proof}
    Let $A=\set{a,b,c,d}$, $|V|=1$. Have $F$ behave as follows:
    \begin{enumerate}
    \item
    If the voter's first two choices are $\set{a,b},\set{a,c},\set{b,c}$ or $\set{b,d}$, the voter gets his first two choices.
    \item
    If the voter's first two choices are $\set{a,d}$, he gets his first and third choice.
    \end{enumerate}
    
    For intuition, suppose $\set{b,c}$ are apples, $\set{a,d}$ oranges,
    and the voter is asked to pick a pair that contains at least one apple.
    
    The rule $F$ clearly satisfies weak viability. To see that it is strategy-proof,
    observe that the only potentially manipulable profiles involve the voter
    ranking $a$ first and $d$ second, or $d$ first and $a$ second -- in 
    every other case the voter already gets his top two choices elected.
    
    Without loss of generality, suppose the voter's preference order is
    $a\succ_1 d\succ_1 b\succ_1 c$. The current election outcome is
    $\set{a,b}$. The voter cannot manipulate optimistically, because
    he already gets his top choice elected. The only member of $S_2(A)$
    that is pessimistically preferred to $\set{a,b}$ is $\set{a,d}$,
    and that is not in the range of the election, so the voter cannot
    manipulate pessimistically either.
\end{proof}
\begin{corollary}
    There exist irreducible consular election rules satisfying SPP, SPO and weak viability that do not have a strong dictator.
\end{corollary}

It turns out, then, that the range of a strategy-proof election can be a non-trivial
affair, and we need to understand this range if we are to have any
hope of determining how these elections behave.

The key observation in this section is that the range of a consular election rule can be visualised
as a graph, with the alternatives playing the role of vertices and an edge existing between
$a$ and $b$ if and only if $\set{a,b}$ is a committee in the range of the funcion.

\begin{definition}
    The \emph{range graph} of a consular election rule $F$, $\CG(F)$, is a graph with vertex set $A$ and edge set $Range(F)$.
\end{definition}

\begin{fact}
    Let $F$ be a consular election rule. Then $F$ is weakly viable if and only if $\CG(F)$ has no isolated vertices, and $F$ is onto if and only if $\CG(F)$ is a complete graph.
\end{fact}

Reducible elections are easily interpreted in this setting.

\begin{fact}
    If $F$ is reducible then $\CG(F)$ is bipartite. If $\CG(F)$ is bipartite and $F$ satisfies SPP and SPO, then $F$ is reducible.
\end{fact}
\begin{proof}
    The first is obvious -- $A$ can be partitioned into $B$ and $C$ such that social choice functions $G$ and $H$ select elements from $B$ and $C$ respectively. As such, there can be no edge within $B$ or within $C$.
    
    For the second part, let $B$ and $C$, $B\uplus C=A$, partition $A$ into the two halves  of $\CG(F)$. Define $G:\CP|B\ria B$ to be the following social choice function:
    $$G(P)=F(P^C) | B,$$
    where $P^C$ is obtained from $P$ by appending the elements from $C$ in lexicographic order at the end of every ballot in $P$. That is, $G$ completes a profile over $B$ into a profile over $A$ by adding $C$ at the end of the ballots, computes $F$, and returns the element of $B$ chosen by $F$. Define  $H:\CP|C\ria C$ analogously. We claim that :
    $$F(P)=\set{G(P|B),H(P|C)}.$$
    Suppose, for contradiction, that there exists a $P$ such that $F(P)\neq\set{G(P|B),H(P|C)}$. Without loss of generality, let $G$ be the culprit. That is, $F(P)=\set{x\in B, y\in C}$, $x\neq G(P|B)$. Modify $P$ by moving the elements in $C$ until they are at the bottom of every ballot, in lexicographical order. By monotonocity, this should only change $y$. However, now we have a profile that must force $F$ to agree with $G$.
\end{proof}
\begin{corollary}\label{cor:bipartite}
    Let $F$ satisfy SPP and SPO. Then $F$ is reducible if and only if $\CG(F)$ is bipartite.
\end{corollary}

\begin{example}
    In general (without strategy-proofness) $F$ can be irreducible and $\CG(F)$ bipartite. Consider $F$ with
    $A=\set{a,b,c,d}$ and one voter, defined as follows:
    \begin{enumerate}
        \item
            If the voter prefers $c$ to $d$, then $a$ is on the committee. Else, $b$ is
            on the committee.
        \item
            If the voter prefers $a$ to $b$, then $c$ is on the committee. Else,
            $d$ is on the committee.
    \end{enumerate}
    The possible committees are $\set{a,c},\set{a,d},\set{b,c}$ and $\set{b,d}$.
    As such, $\CG(F)$ is bipartite. However, $F$ is irreducible -- the decision
    on whether to include $a$ on the committee is contingent on the relative
    ranking of $c$ and $d$, so there can be no $G:\CP(V,\set{a,b})\ria\set{a,b}$
    that makes the decision.
\end{example}

\begin{fact}
    Let $F$ be a weakly viable reducible consular election rule satisfying SPP and SPO. 
    Then $\CG(F)$ is a complete bipartite graph. 
\end{fact}
\begin{proof}
    Let $A=B\uplus C$, the two halves of a bipartite graph. Suppose for contradiction 
    that there exist $x\in B,y\in C$ such that $x$ and $y$ are not isolated, 
    but $(x,y)\notin\CG(F)$. Let $F(P)=\set{x,z}$. Move $y$ to the top of each ballot. 
    By unanimity, $y$ must be a winner. By reducibility, it must replace $z$.
\end{proof}
\begin{corollary}\label{cor:Marianacyclic}
    Let $F$ be a weakly viable reducible consular election rule satisfying SPP and SPO.
    Then $F$ is a Marian election if and only if $\CG(F)$ is acyclic.
\end{corollary}
\begin{proof}
    Irreducible elections are non-bipartite, hence they contain an odd cycle. The
    only complete bipartite graph that is acyclic is a tree of depth 1.
\end{proof}

Now that we know what the range of a reducible election looks like,
we of course know exactly what the strategy-proof irreducible elections are: simply
those $F$ for which $\CG(F)$ has an odd cycle.
This does not give us a lot to work with, so we must look further.
The proof above relied on the fact that strategy-proof elections are unanimous, which
allows us to prove that certain edges in the range graph must exist by considering
the profiles where all voters rank some alternative first. Rather than
using strategy-proofness in an ad-hoc way through this section,
we can use it to
prove a much stronger property about the range graphs of
strategy-proof elections.

\begin{definition}
    We say a graph $G$ satisfies \emph{edge-connectivity} if, whenever
    $\set{a,b}$ and $\set{c,d}$ are non-incident edges,
    we have:
    \begin{enumerate}
        \item 
            $\set{a,c}$ or $\set{a,d}$ is an edge.
        \item
            $\set{b,c}$ or $\set{b,d}$ is an edge.
    \end{enumerate}
    Note that the relation is symmetric. Namely, if $\set{a,b}$
    and $\set{c,d}$ are non-incident edges, then so are $\set{c,d}$
    and $\set{a,b}$. Thus, edge connectivity also implies that:
    \begin{enumerate}
        \item 
            $\set{a,c}$ or $\set{b,c}$ is an edge.
        \item
            $\set{a,d}$ or $\set{b,d}$ is an edge.
    \end{enumerate}
    
    It is easier to visualise what this property entails by observing that
    a graph is edge-connected if and only if it does
    not contain any of the following as induced subgraphs:
    
    \begin{center}
    \begin{tabular}{c c c}
    \begin{tikzpicture}[auto,
                        thick]
    
      \node (a) {$1$};
      \node (b) [right of=a] {$2$};
      \node (c) [below of=a] {$3$};
      \node (d) [below of=b] {$4$};
    
      \path[every node/.style={font=\sffamily\small}]
        (a) edge node {} (b)
        (c) edge node {} (d)
        ;
    \end{tikzpicture}
    &
    \begin{tikzpicture}[auto,
                        thick]
    
      \node (a) {$1$};
      \node (b) [right of=a] {$2$};
      \node (c) [below of=a] {$3$};
      \node (d) [below of=b] {$4$};
    
      \path[every node/.style={font=\sffamily\small}]
        (a) edge node {} (b)
        (b) edge node {} (d)
        (c) edge node {} (d)
        ;
    \end{tikzpicture}
    &
    \begin{tikzpicture}[auto,
                        thick]
    
      \node (a) {$1$};
      \node (b) [right of=a] {$2$};
      \node (c) [below of=a] {$3$};
      \node (d) [below of=b] {$4$};
    
      \path[every node/.style={font=\sffamily\small}]
        (a) edge node {} (b)
        (b) edge node {} (c)
        (c) edge node {} (d)
        (b) edge node {} (d)
        ;
    \end{tikzpicture}
    \end{tabular}
    \end{center}

\end{definition}

\begin{example}
    The Lex Licinia Sextia was a series of laws passed in 367 BC to address the disproportionate
    dominance patricians enjoyed in the Roman state. Under the new laws
    at least one of the two consuls of the Roman republic was
    required to be plebeian. Thus the valid configurations of the consulship were: a patrician
    and a plebeian, or two plebs (Jupiter forbid). The range graph of this election would thus
    consist of a clique $X$ of plebeians, and a totally disconnected set $Y$ of patricians with the property that: for $x\in X$ and $y\in Y$, $\set{x,y}$ is an edge.
    This graph satisfies edge-connectivity. 
    Other families of edge-connected graphs include the complete and complete bipartite graphs.   
    In \cref{fig:edgeconnected} we illustrate all the edge-connected
    graphs of orders 4, 5 and 6 that are neither complete nor
    bipartite. This includes the plebeian/patrician family as well
    as others.
\end{example}

\begin{figure}
    \begin{tabular}{ c c c}
         &\includegraphics[scale=0.45]{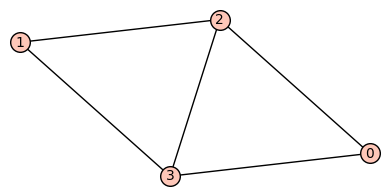}&\\
         \includegraphics[scale=0.45]{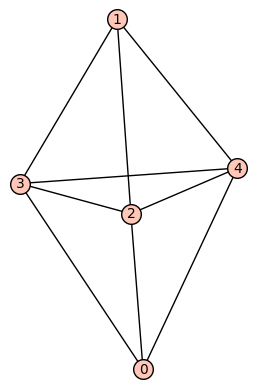}&\includegraphics[scale=0.45]{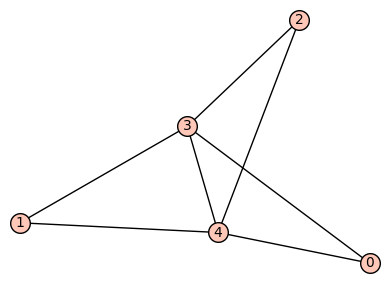}&\includegraphics[scale=0.45]{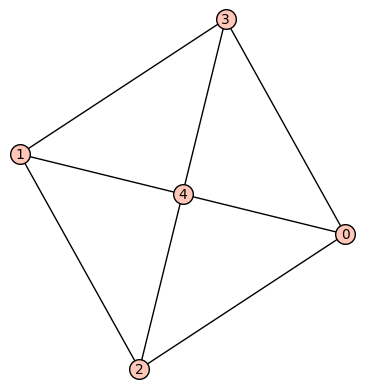}\\
         \includegraphics[scale=0.45]{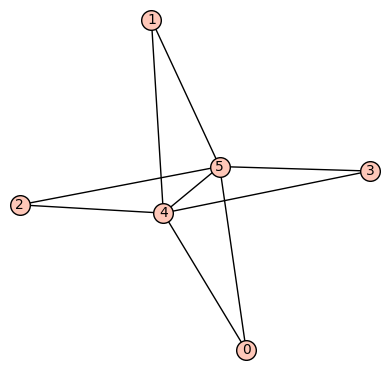}&\includegraphics[scale=0.45]{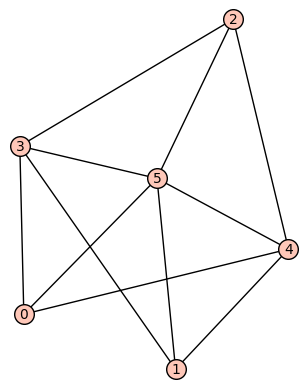}&\includegraphics[scale=0.45]{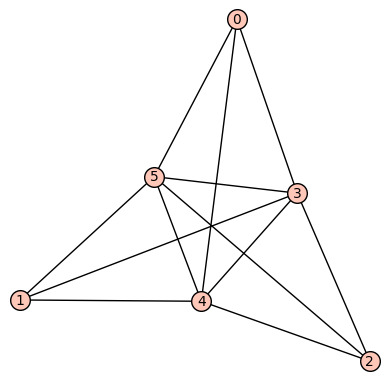}\\
         \includegraphics[scale=0.45]{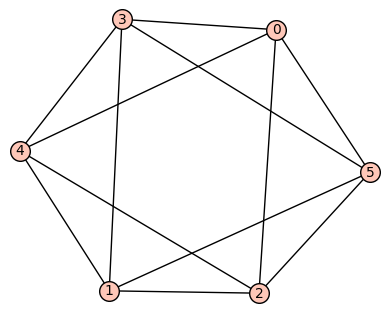}&\includegraphics[scale=0.45]{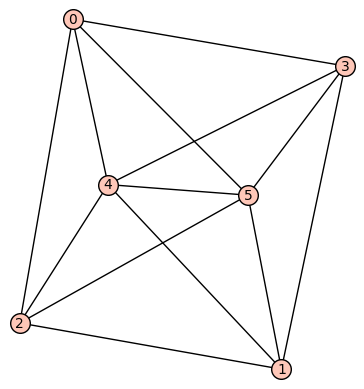}&\includegraphics[scale=0.45]{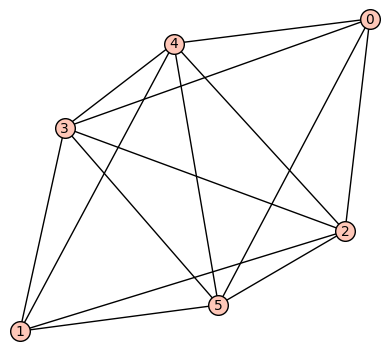}
    \end{tabular}
    
    \caption{All edge-connected graphs of orders 4, 5 and 6 that are neither
    complete nor bipartite.}\label{fig:edgeconnected}
\end{figure}

Edge-connectivity is not merely a consequence of strategy-proofness,
it is the closest we can get to a characterisation of the property
if we focus on the range of an election alone. Crucial to this
is the fact that if the range of an election is edge-connected,
then any given voter will be able to identify a unique ``best'' committee:
a pair which is neither pessimistically nor optimistically dominated
by any other pair in the range.

\begin{definition}
    Consider an $X\sub S_2(A)$, and a voter $i$. Voter $i$'s
    \emph{favourite committee} in $X$, if it exists, is an element
    that is maximal under $\succeq_i^P$ and $\succeq_i^O$ over $X$.
\end{definition}

\begin{proposition}\label{fact:favcommittee}
    Suppose $\CG(F)$ is edge-connected.
    Then each voter has a unique favourite committee in $\CG(F)$.
\end{proposition}
\begin{proof}
Let $G$ be an edge-connected graph with at least one edge. 
    Consider a voter $i$ with preferences $P_i$, and the
    induced optimistic and pessimistic orders $\succeq_i^O$
    and $\succeq_i^P$.
    
    Let $\set{a,b}$ be a maximal
    element under $\succeq_i^O$. Without loss of generality, suppose $a\succ_i b$. We will show that
    one of the following is true:
    \begin{enumerate}
        \item $\set{a,b}$ is maximal under $\succeq_i^P$, or
        \item There exists a $\set{a,c}$
        such that $\set{a,c}\succ_i^P\set{a,b}$
        and $\set{a,c}$ is maximal under $\succeq_i^O$.
    \end{enumerate}
    Since the graph in question is finite, repeatedly applying 2 will eventually yield us
    1 -- the existence of an element maximal in
    both orders.
    
    Suppose $\set{a,b}$ is not
    maximal under $\succeq_i^P$. There must exist
    some $X\succ_i^P\set{a,b}$. If $X$ is of the form
    $\set{a,c}$, we are done. Suppose $X=\set{c,d}$,
    for $c,d$ distinct from $a$. Observe that this implies that
    $c,d\succ_i b$.  We have assumed that $\set{a,b}$ and $\set{c,d}$
    are both edges. By edge-connectivity, either
     $\set{a,c}$ or $\set{a,d}$ is an edge,
     and both of these pessimistically dominate
     $\set{a,b}$.
    
    This establishes that there exists a $\set{a,b}$
    that is maximal in both the pessimistic and the
    optimistic order. To see that it is unique, observe simply that an edge maximal optimistically must include $a$, else $\set{a,b}$ would dominate, and likewise
    an edge maximal pessimistically must include
    $b$.
\end{proof}

Edge-connectedness is strongly connected to strategy-proofness for consular election rules.

\begin{proposition}\label{fact:edgeconnected}
    If $F$ satisfies SPP and SPO then $\CG(F)$ satisfies edge-connectivity. Conversely if
 $G$ is a graph with at least one edge that satisfies edge-connectivity, then there exists an $F$
    satisfying SPP and SPO such that $\CG(F)=G$.
\end{proposition}
\begin{proof}
    To show that strategy-proofness implies edge-connectivity, take an $F$ satisfying SPP and SPO and let $F(P)=\set{a,b}$. Let $\set{c,d}$ be in the range of
    $F$. Move $c$ to the top of $P$. By unanimity and upwards
    monotonicity, $F(P)=\set{a,c}$ or $\set{b,c}$. Repeat for $d$. 
    
    For the other direction, given an edge-connected $G$, fix a voter $i$ and let $F$ be the function that 
    gives $i$ the unique edge of $G$ that is maximal under both $\succeq_i^P$ and $\succeq_i^O$. First,
    observe that $\CG(F)=G$ -- given an edge $\set{a,b}$, construct a profile $P$
    where $i$ ranks $a$ and $b$ first. $F(P)=\set{a,b}$, so the range of $F$
    is precisely the edges of $G$.
    That $F$ satisfies SPP and SPO follows from the fact that the $i$ is getting his favourite committee -- an element maximal under both the pessimistic and the optimistic order -- from which no manipulation is possible.
\end{proof}
\begin{corollary}
\label{cor:favourite}
Let $F$ be a consular election rule satisfying SPP and SPO. Then each voter has a unique favourite committee in $\CG(F)$.
\end{corollary}
\begin{corollary}
    Let $F$ be a consular election satisfying weak viability, SPP and SPO. $\CG(F)$ has diameter $2$.
\end{corollary}
\begin{proof}
 Let $a$ and $b$ be arbitrary vertices in $\CG(F)$.
 By weak viability, there exist vertices $c$ and $d$, not necessarily distinct, for which $\set{a,c}$
 and $\set{b,d}$ are edges. If $c=b$ or $d=a$, then
 the distance between $a$ and $b$ is one. If $c=d$,
 then the distance between $a$ and $b$ is at most two.
 
 Suppose then that $c$ and $d$ are distinct from $a$, $b$, and
 each other. By edge-connectivity, at least one of $\set{a,b}$ and $\set{a,d}$ must be an edge. If it is $\set{a,b}$, the distance between $a$ and $b$ is one. If it is $\set{a,d}$, then the distance between $a$ and $b$ is at most two.\end{proof}

Naturally, this introduces a new notion of dictatorship:

\begin{definition}\label{def:rangedictator}
    A \emph{range dictator} is a voter such that for every profile, that voter's 
    favourite committee is elected.
\end{definition}

Under the assumption of weak viability, a range dictator is a stronger notion than a weak dictator, and
collapses to a strong dictator in the case of an onto election.

Edge-connectivity is a very strong property, and it allows us to use
the fact that an irreducible election has an odd cycle to prove
a much stronger statement about its structure.

\begin{fact}\label{fact:3cycles}
    If $F$ is irreducible, SPP, SPO and weakly viable then every vertex in $\CG(F)$ belongs to a $3$-cycle.
\end{fact}
\begin{proof}
    $\CG(F)$ is not bipartite, so an odd length cycle exists. Suppose the smallest such cycle
    is of length $2k+1$, $k>1$. Let $\set{a,b},\set{b,c},\set{c,d}$ be edges of this cycle. 
    By \cref{fact:edgeconnected} either $\set{a,d}$ or $\set{b,d}$ is an edge. 
    If it is $\set{a,d}$, then we have a $2k-1$ cycle by replacing
    $\set{a,b},\set{b,c},\set{c,d}$ with $\set{a,d}$. 
    If it is $\set{b,d}$, then we have a $3$-cycle $\set{b,c},\set{c,d},\set{b,d}$. 
    Either way, $2k+1$ is not the smallest length of an odd cycle.
    
    Given that a $3$-cycle exists, assume there exists an $a$ that does not belong to a $3$-cycle.
    By connectivity (diameter 2), we can choose $a$ such that $a$ is adjacent to $b$, 
    and $b,c,d$ is a $3$-cycle. Let $F(P)=\set{c,d}$. Move $a$ to the top of the ballots. 
    The winner is either $\set{a,c}$ or $\set{a,d}$, giving the $3$-cycle $(a,b,c)$ or $(a,b,d)$.
\end{proof}

This leads to the following corollary, which will prove crucial
in the proof of \cref{prop:irreduciblelinked}.

\begin{corollary}\label{cor:triangle}
    Let $F$ be an irreducible election satisfying SPP, SPO and weak viability. If $a$ and $b$ are connected by an edge, they also form a $3$-cycle with some $c$.
\end{corollary}
\begin{proof}
    By \cref{fact:3cycles}, $a$ belongs to some $3$-cycle. If this cycle contains $b$, we are done.
    Else, let this cycle be $(a,c_1,c_2)$. By edge connectivity, either $\set{b,c_1}$ or $\set{b,c_2}$
    is an edge, so that $(b,c_1, a)$ or $(b,c_2,a)$ is a cycle.
\end{proof}

\section{Irreducible elections and linked domains}
\label{s:linked}

We are finally in a position to state exactly what we are trying
to prove: an irreducible consular election rule that is manipulable
by neither an optimist nor a pessimist has a range dictator. We do
not even have to assume that $|A|\geq 3$ -- the only consular
election rule over two alternatives is both reducible and dictatorial.

While a direct proof of the theorem is possible using the standard
techniques of social choice theory, a shorter and more elegant
solution involves following the footsteps of \citet{Ozyurt2008} and \citet{Reffgen2011} by using 
a theorem
of \citet{Aswal2003}. We define an extension from 
preferences over $A$ to preferences over $S_2(A)$, and apply the result 
of \citet{Aswal2003} to show that the resulting social choice
function is dictatorial.

\begin{definition}\label{def:connected}
    Let $D$ be a domain of linear orders.
    Two alternatives, $a,b\in A$, are \emph{connected} in $D$
    if there exist $P,P'\in D$ such that in $P$, $a$ is ranked first and $b$
    second, while in $P'$, $b$ is ranked first, and $a$ second.
\end{definition}

\begin{definition}\label{def:linked}
    A domain $D$ is said to be \emph{linked} if it is possible
    to order $A=\set{a_1,\dots,a_m}$ in a way that:
    \begin{enumerate}
        \item $a_1$ is connected to $a_2$.
        \item For $i\geq 3$, $a_i$ is connected to at least two elements
        in $\set{a_1,\dots,a_{i-1}}$.
    \end{enumerate}
\end{definition}

\begin{theorem}[\cite{Aswal2003}]\label{thm:aswal}
    Let $D$ be a linked domain of preferences over $A$, $|A|\geq 3$,
    and $F:D^n\ria A$ be a unanimous social choice function. Then $F$ is strategy-proof
    if and only if $F$ is dictatorial.
\end{theorem}

The plan is, given an $F:\CP(V,A)\ria S_2(A)$, to define an extension
map $\alpha$ that transforms a linear order over $A$ into a linear
order over $S_2(A)$. We then define the social choice function
$\CF:\alpha(\CP(V,A))\ria S_2(A)$ by $\CF(\alpha(P))=F(P)$. If we
show that $\alpha(\CP(V,A))$ is linked and $\CF$ is unanimous, we
can invoke \cref{thm:aswal} to show that $\CF$ is dictatorial.
Since $\CF$ mirrors the output of $F$, this must mean that $F$ is
dictatorial as well -- provided, of course, that $\alpha$ is defined
in a way voter $i$'s first choice in $\alpha(P_i)$ corresponds to
voter $i$'s favourite committee in $S_2(A)$.

The first issue we face is in extending the preferences. As we have previously argued, Duggan-Schwartz manipulation, which can be expressed as $\succ_i^O\cup\succ_i^P$, does not define
a transitive order.\footnote{Among other problems. If we take $\succ_i^O\cup\succ_i^P$ to represent a voter's strict preferences, we have
to posit that a voter is able to at once strictly prefer $X$ to $Y$ and
$Y$ to $X$ -- what is sometimes known as \emph{conflicted} preferences
in the literature.} What we can do instead is to compose the two lexicographically. That is,
define $\succeq_i^L$ by:
$$X\succeq_i^L Y \text{ iff } \begin{cases}X\succ_i^O Y, \text{ or}\\
X\succeq_i^OY \text{ and } X\succeq_i^PY\end{cases}.$$
That is,  the voter ranks committees by the best element first, and then
uses the worst element as a tie-breaker. This clearly yields a linear
order.

Of course, manipulation by $\succeq_i^L$ is not the same thing as manipulation
by an optimist or a pessimist. However, in the case of a consular election rule, lexicographic manipulation is
a strictly narrower notion. As a consequence, a function that satisfies SPP and SPO also satisfies $\succeq_i^L$-strategy-proofness. If we thus show that $\succeq_i^L$-strategy-proofness
is enough to ensure dictatoriality, we will \emph{a fortiori} show that elections
that satisfy SPP and SPO are dictatorial.
\begin{fact}
    A consular election rule $F$ is $\succ_i^L$-strategy-proof only if it satisfies SPP and SPO.
\end{fact}

Our next problem is that, in general, $\CF$ is not going to be
unanimous -- unanimity would imply that $\CF$ is onto, and we have argued
that many natural irreducible elections fail that property. To address this,
we will consider the restriction of $\CF$ to the relevant committees.
\begin{definition}
    Let $F:\CP(V,A)\ria S_2(A)$ be a weakly viable consular election rule
    satisfying SPP and SPO, and let $\CF:\alpha(\CP(V,A))\ria S_2(A)$
    be the induced social choice function.
    
    Define $\CF':(\alpha(\CP(V,A))|\CG(F))\ria (S_2(A)|\CG(F))$ by
    $\CF'(\alpha(P)|\CG(F))=\CF(\alpha(P))$.
\end{definition}

At first glance,
    $\CF'$ may appear ill-defined: there may exist $\alpha(P)\neq \alpha(P')$ such that
    $\alpha(P)|\CG(F)=\alpha(P')|\CG(F)$. However, this is not the case -- $\alpha(P)$ is the lexicographic order over pairs defined by the linear orders in $P$. As such, no matter how many pairs we remove from $\alpha(P)$ to obtain $\alpha(P)|\CG(F)$, there remains a unique way to reconstruct the order.

\begin{example}
    The dense notation of the above statement can hide the fact that
    we are talking about a very simple operation. Let $F$ be the function
    from \cref{fact:4notonto} -- one voter is asked to pick a pair
    from $\set{a,b,c,d}$, at least one element of which is from $\set{b,c}$.
    
    Consider a voter with preference $a\succ_1 d\succ_1 b\succ_1 c$. The extension
    $\alpha$ turns his preferences into:
    $$\set{a,d}\succ_1\set{a,b}\succ_1\set{a,c}\succ_1\set{b,d}\succ_1\set{c,d}\succ_1\set{b,c}.$$
    The function $\CF$ is then the function that gives voter 1 his first choice,
    unless it happens to be $\set{a,d}$, in which case it gives him his second choice.
    $\CF$ is not unanimous, as it is not onto, and unsurprisingly $\CF$ is not dictatorial.
    
    The graph $\CG(F)$ does not include $\set{a,d}$ as an edge. If we restrict
    voter 1's preferences to $\CG(F)$, then we get $\alpha(P)|\CG(F)$:
    $$\set{a,b}\succ_1\set{a,c}\succ_1\set{b,d}\succ_1\set{c,d}\succ_1\set{b,c}.$$
    Now, $\CF'$ is simply the function that gives voter 1 his first choice.
    This function is unanimous and dictatorial. 
\end{example}

That the function in the example above is unanimous
is not an accident.
\begin{lemma}
Let $F$ be a consular election satisfying SPP and SPO. Then voter $i$'s first choice in $\alpha(P_i)|\CG(F)$
is his favourite committee.
\end{lemma}
\begin{proof}
    Let $\set{a,b}$ be voter $i$'s first choice in $\alpha(P_i)|\CG(F)$. Recall that the favourite committee is an element
    maximal under both $\succeq_i^O$ and $\succeq_i^P$. That the $\set{a,b}$ is maximal under $\succeq_i^O$ is obvious -- $\succeq_i^L$
    ranks pairs by the best element first. Without loss of generality, let the best element be $a$.
    
    In \cref{fact:favcommittee}, we have established
    that voter $i$ has a unique favourite committee in $\CG(F)$. In order to be maximal under $\succeq_i^O$, the favourite committee must contain $a$. However, it cannot be some $\set{a,c\neq b}$; if $b\succ_i c$, then
    the favourite committee is not maximal under $\succeq_i^P$, and if $c\succ_i b$ then
    $\set{a,b}$ is not the first choice in $\alpha(P_i)|\CG(F)$. Thus the favourite committee and the first choice in $\alpha(P_i)|\CG(F)$ must coincide.
\end{proof}
\begin{lemma}\label{lem:strongunanimous}
    Let $F$ be a consular election satisfying SPP, SPO. The pair $\set{a,b}$ is elected whenever $\set{a,b}$ is the favourite committee of every voter.
\end{lemma}
\begin{proof}
    By \cref{fact:favcommittee}, $\set{a,b}$
    is in the range of $F$. Let $P'$ be a profile
    such that $F(P')=\set{a,b}$. Move $a$ and $b$
    to the top of every ballot. By monotonicity,
    the winning committee is still $\set{a,b}$.
    Now move the other alternatives, voter by voter,
    to the positions they occupy in $P$. By monotonicity, if the winning committee changes during this process then that must mean some $c$
    overtook, without loss of generality, $a$ and
    the winning committee became $\set{b,c}$. However, this must mean that some voter prefers
    $\set{b,c}$ to $\set{a,b}$, which contradicts
    $\set{a,b}$ being the favourite committee.
\end{proof}
\begin{corollary}\label{lem:restrictionunanimous}
    Let $F:\CP(V,A)\ria S_2(A)$ be a weakly viable consular election rule
    satisfying SPP and SPO.     
    Then $\CF':(\alpha(\CP(V,A))|\CG(F))\ria (S_2(A)|\CG(F))$ is unanimous.
\end{corollary}

We now have a unanimous function. We next consider whether its domain
is linked. The graph interpretation is of use here, as it turns out
that connected committees are exactly the incident edges.

\begin{lemma}\label{lem:connect}
Let $F:\CP(V,A)\ria S_2(A)$ be a weakly viable consular election rule
    satisfying SPP and SPO. Then two committees, $X,Y\in S_2(A)$, $X\neq Y$, are connected in $\alpha(\CP(V,A))|\CG(F)$
    if and only if the edges $X$ and $Y$ are incident in $\CG(F)$.
\end{lemma}
\begin{proof}
    Let $X=\set{a,b}$ and $Y=\set{b,c}$. Let $P$ be a linear order
    ranking $b\succ a\succ c\succ\dots$ and $P'$ ranking $b\succ c\succ a\succ\dots$. Recall that $\alpha(P)|\CG(F)$ is the lexicographic order
    $\succeq_i^L$. The first element is clearly
    $\set{a,b}$ and the second $\set{b,c}$. Likewise for $\alpha(P')|\CG(F)$,
    the first element is $\set{b,c}$ and the second $\set{a,b}$.
    
    Now suppose $X=\set{a,b}$ and $Y=\set{c,d}$, all $a,b,c,d$ distinct.
    Let $P$ be any preference order for which the first element in $\alpha(P)|\CG(F)$ is $\set{a,b}$. As $\alpha(P)|\CG(F)$ is the lexicographic order, either $a$ or $b$ have to be ranked first.
    Without loss of generality, let it be $a$. Obseve that one
    of the following must be true:
    \begin{enumerate}
        \item The second element in $\alpha(P)|\CG(F)$ is $\set{a,e}$
        for some $e$.
        \item $\set{a,b}$ is the only edge in $\CG(F)$ incident on $a$.
    \end{enumerate}
    If the first is true, $\set{c,d}$ cannot be ranked second, and
    hence $X$ and $Y$ cannot be connected. If the second is true,
    we violate edge connectivity -- $\set{a,b}$ and $\set{c,d}$
    are both edges, so either $\set{a,c}$ or $\set{a,d}$ must
    be an edge.
\end{proof}

\begin{proposition}\label{prop:irreduciblelinked}
    Let $F:\CP(V,A)\ria S_2(A)$ be an irreducible, weakly viable consular election rule
    satisfying SPP and SPO. Then $\alpha(\CP(V,A))|\CG(F)$ is linked.
\end{proposition}
\begin{proof}
    Applying \cref{lem:connect} to \cref{def:linked}, we want
    to show that it is possible to order the edges $X_1,\dots,X_q$
    of $\CG(F)$ in a way that:
    \begin{enumerate}
        \item $X_1$ is incident on $X_2$.
        \item For $i\geq 3$, $X_i$ is incident on at least two elements
        of $\set{X_1,\dots,X_{i-1}}$.
    \end{enumerate}
    To obtain 1, observe that irreducibility implies $|A|\geq 3$,
    and weak viability implies that no vertex is isolated. To connect
    three vertices we need at least two edges, call them $\set{a,b}$
    and $\set{c,d}$. If $a,b,c,d$ are not distinct, the two edges
    are incident. If $a,b,c,d$ are distinct, edge connectivity
    implies either $\set{a,c}$ or $\set{a,d}$ is an edge, both of
    which are incident on $\set{a,b}$.
    
    To obtain 2, we will demonstrate that given a set of edges $I\subset\CG(F)$
    which satisfies the property that every that every $e\in I$
    is incident on at least one other member of $I$ (i.e. $I$ has no
    isolated edges). we can find a $e'\notin I$ such that $e'$ is incident
    on at least two members of $I$. The reader will note that this
    means that $I\cup\set{e'}$ also satisfies the property that every
    $e\in I\cup\set{e'}$ is incident on at least one other member of
    $I\cup\set{e'}$. As such, starting with a set of two incident
    edges, it is possible to continually add edges until the entire
    graph is shown to satisfy 2.
    
    Let $I$ be such a set. As the graph is connected, there
    must exist a $\set{x,y}\notin I$ such that
    $\set{x,y}$ is incident on $I$. Without loss of generality, suppose
    $x$ belongs to an edge in $I$. If $y$ also belongs to an edge in $I$
    we are done -- $\set{x,y}$ is incident on $\set{x,v}$ and $\set{y,w}$
    for some $\set{x,v},\set{y,w}\in I$. Suppose then that $y$ does not
    belong to an edge in $I$.
    
    Let $\set{x,w}$ be an edge in $I$. Recall that $\set{x,w}$ is
    incident on at least one member of $I$. If that member is
    $\set{x,v}$, we are done -- $\set{x,y}$ is incident on $\set{x,w}$
    and $\set{x,v}$, and we can add it to $I$. If that member is
    $\set{v,w}$, we can visualise the situation as follows, taking
    the vertices at the top to belong to $I$.
    
    \begin{center}
    \begin{tikzpicture}[auto,
                        thick]
    
      \node (a) {$x$};
      \node (b) [right of=a] {$v$};
      \node (c) [right of=b] {$w$};
      \node (d) [below of=a] {$y$};
    
      \path[every node/.style={font=\sffamily\small}]
        (a) edge node {} (d)
        (a) edge node {} (b)
        (b) edge node {} (c)
        ;
    \end{tikzpicture}
    \end{center}
    
    By \cref{cor:triangle}, there exists a $z$ such that
    $\set{x,z}$ and $\set{y,z}$ are both edges. If $z=v$
    we are done -- $\set{y,v}$ is incident on two
    edges in $I$, so we can add it to the set.
    If $z\neq v$ belongs to an edge in $I$ we are also done,
    as $\set{x,y}$ is incident on $\set{x,z}$ and $\set{x,v}$.
    Suppose then that $z$ does not belong to an edge in $I$.
    
    \begin{center}
    \begin{tikzpicture}[auto,
                        thick]
    
      \node (a) {$x$};
      \node (b) [right of=a] {$v$};
      \node (c) [right of=b] {$w$};
      \node (d) [below of=a] {$y$};
      \node (e) [right of=d] {$z$};
    
      \path[every node/.style={font=\sffamily\small}]
        (a) edge node {} (d)
        (a) edge node {} (b)
        (b) edge node {} (c)
        (e) edge node {} (a)
        (e) edge node {} (d)
        ;
    \end{tikzpicture}
    \end{center}
    If no other edges are present, this is a forbidden induced subgraph. 
    We must posit that either $\set{y,v}$ or
    $\set{v,z}$ is an edge. If it is $\set{y,v}$, we
    can add $\set{y,v}$, incident on $\set{x,v}$ and $\set{v,w}$, to $I$. If it is $\set{v,z}$,
    we can add $\set{v,z}$, incident on $\set{x,z}$ and $\set{x,v}$, to $I$. Either way we increase the size of $I$.
\end{proof}

We finally achieve the promised main result.

\begin{corollary}
\label{cor:main}
    $\CF':\alpha(\CP(V,A))\ria S_2(A)$ is dictatorial,
    hence $F$ has a range dictator.
\end{corollary}
\begin{proof}
By \cref{prop:irreduciblelinked} the domain is linked, and by
\cref{lem:restrictionunanimous} the function is unanimous. We can apply \cref{thm:aswal}.
\end{proof}

As a sanity check, we can verify that Marian elections also have a linked domain,
whereas non-Marian reducible elections do not.

\begin{proposition}
    Let $F:\CP(V,A)\ria S_2(A)$ be a Marian, weakly viable consular election rule
    satisfying SPP and SPO.   
   Then $\alpha(\CP(V,A))|\CG(F)$ is linked.
\end{proposition}
\begin{proof}
    Every edge is incident on each other.
\end{proof}
This gives us \cref{prop:Mariandictator}. 

\begin{corollary}
    For $|A|\geq 4$, a Marian, weakly viable election
    satisfying SPP and SPO has a weak/range dictator (the two notions coincide).
\end{corollary}
\begin{proof}
    For $|A|\geq 4$, the range graph of the election has at least three edges. This means that $\alpha(\CP(V,A))|\CG(F)$ is a linked domain over the set of alternatives $S_2(A)|\CG(F)$, with
    $|S_2(A)|\CG(F)|\geq 3$, so we can invoke \cref{thm:aswal}.
\end{proof}

\begin{proposition}
    Let $F:\CP(V,A)\ria S_2(A)$ be a non-Marian, reducible, weakly viable consular election rule
    satisfying SPP and SPO. Then $\alpha(\CP(V,A))|\CG(F)$ is not linked.
\end{proposition}
\begin{proof}
    Recall that $\CG(F)$ is a complete bipartite
    graph that is not acyclic (i.e., not
    a star graph). Let $X_1=\set{a,b}$ and $X_2=\set{a,c}$ be any two incident edges. Let $V_1$ and $V_2$ be the two halves of $\CG(F)$, $a\in V_1$ and $b,c\in V_2$. Observe
    that the only way to expand $\set{X_1,X_2}$ consistent with
    \cref{def:linked} is to add some $\set{a,d}$ to the set. As
    $\CG(F)$ is not a star graph, by doing so we can never cover
    the entire graph.
\end{proof}

\section{Summary}

We have demonstrated that the condition of weak viability allows for the existence of non-trivial
structure in the range of an election, which
we have visualised as a property of the resulting
range graphs.

For reducible elections, we have shown the following:
\begin{itemize}
\item
If the election is Marian, for $|A|\geq 4$ a range dictator exists.
\item
If the election is non-Marian then a dictator may not exist. However, for $|A|\geq 5$ a ``partial dictator'' exists, in the sense that some voter
can unilaterally decide which element of $B\subset A$ is present on the winning committee.
\end{itemize}
For irreducible elections we have a dictatoriality result:
\begin{itemize}
\item
For $|A|\geq 4$, a range dictator exists.
\end{itemize}

\subsection{Discussion}

\subsubsection{Extending the results}

According to Livy, the last king of Rome was overthrown in 509 BC by a body of citizens
led by Lucius Junius Brutus. The first act of Brutus after the expulsion of the king
was to have the people swear an oath to never let a single man govern Rome. It is
on this basis that the Romans justified the duplication of the post of consul and
other key positions in government.

Like most Roman political experiments, this one eventually ended in failure. Two is
not a good number in government. Any disagreement results in an even split, which means that
to avoid political deadlock, elaborate schemes of power sharing need to be devised. One
such scheme used by the Romans involved the consuls commanding the army on alternating days --
this led to the farcical prelude to the Battle of Cannae where Varro ordered the army
to advance during his days of command, and Paullus held it back during his. The purpose
of this historical aside is to suggest that while the case of $k=2$ raises
many interesting issues, it should nevertheless be seen as a stepping stone
to studying larger committee sizes, rather than an end in itself.

The argument as presented here fails to generalise fairly quickly -- \cref{lem:sbot2}
already relies on the fact that the winning committee is of size two. Without
\cref{lem:sbot2} there is no monotonicity, and without monotonicity there is
nothing to be said. Admittedly, we did not spend a great deal of time
attempting to generalise these lemmas to arbitrary $k$, and intuitively we suspect
some version of monotonicity for general committee selection rules should not be
too difficult to establish. The more serious issue arises when dealing with the
range of such elections. In the case of $k=2$, we were abetted greatly by the fact
that we could visualise the range of an election as a graph, and use results and intuitions
from graph theory to guide us along. To replicate this approach in the case of $k=3$
would require us to either consider a hypergraph, or to index a range graph by the perspective
of a voter, adding an edge between the best and worst alternative of every committee. It is unclear
that either of these approaches will give us the insights we need, nor is there
any a priori reason to believe that in such a context voters would still
have a favourite committee, without which a range dictator is ill-defined. The case of variable size
committees is even worse: now the hypergraph is not even uniform. Moreover, in the case
of fixed $k$-committees we could potentially side step the issue and study onto rules
only; such a result would be interesting enough in itself. In the case of variable size
committees the onto case is trivialised by the Duggan-Schwartz theorem, so tackling
the range of the election is unavoidable.

The specific technique used in \cref{prop:irreduciblelinked} is also particular to
the case $k=2$ -- in this case the lexicographic order we defined is strictly
narrower than Duggan-Schwartz manipulation and hence a dictatoriality result for
the former implies, a fortiori, the latter. This is no longer the case with $k=3$.
A voter with preferences $a\succ_1 b\succ_1 c\succ_1 d\succ_1 e$ has a pessimistic
manipulation from $\set{a,b,e}$ to $\set{a,c,d}$ but no lexicographic manipulation;
and a lexicographic manipulation from $\set{a,c,d}$ to $\set{a,b,d}$ even though
pessimistically he is worse off. However, we suspect this difficulty to be less
serious than the issue of the range. 

\subsubsection{Committees of arbitrary size}

In recent years, social choice
correspondences have reappeared in the literature
under a number of names: multi-winner voting rules \citep{Meir2008}, resolute social choice correspondences \citep{Ozyurt2008}, irresolute voting rules \citep{handbook2016}, social dichotomy rules \citep{Duddy2014}.  Mathematically, these are the same concept -- a function that takes a tuple of preference orders over
$A$ to a non-empty subset of $A$. What is different is the interpretation. In classical social choice theory, alternatives are taken to mean
mutually exclusive states of the world. A social choice correspondence, then, could take the economic needs of a population to a set of compatible equilibria, or a nation's political preferences to a shortlist of presidential candidates, with the assumption that some tie-breaking procedure will pick one, and only one, of these alternatives at a later date. However, the same mathematical formalism can be used to study an algorithm which aggregates film rankings to produce a library for an aeroplane's entertainment system, or an election that directly picks members of parliament, in which case the outcome of the function is the final result for society, and not a tie to be resolved in the future.

In this paper we have dealt with the special case of rules which must output a committee of fixed size.
If we relax this to allow committees of arbitrary size, then all the impossibility results established for social choice correspondences apply. However, the assumptions underlying these results may be inappropriate for the given interpretation. In particular, the viability assumptions of
the Duggan-Schwartz theorem, requiring that a correspondence be onto with respect to singletons, may be
unjustifiably strong -- there is no \emph{a priori} reason to privilege singletons
over any other committee size. For example, imagine a situation where certain singletons are allowed and others are not -- in a catering scenario it might be the case that a menu consisting of a vegetarian option only is acceptable, but meat only is not. Thus there may still be room in the literature for a new strategy-proofness result applicable to the committee interpretation, with more natural assumptions.

\subsubsection{Manipulation model}

The assumptions of the Duggan-Schwartz theorem make
the most sense in the context of a social choice function
that generally elects singletons, but occasionally produces
a tie which is broken in a non-deterministic manner. In
this setting the requirement that the function be onto
with respect to singletons is a natural analogue of a
social choice; the fact that the value of a committee
is assessed only by the individual values of the alternatives
on that committee makes
sense because after the tie-breaking mechanism is invoked
only one alternative will be selected; and the optimistic
and pessimistic preference orderings used, while extreme,
are nevertheless natural ways to assess the outcome
of the tie.

In the case of a committee selection rule sets do not
represent ties to be resolved, but the final outcome
of the election process. Given that the committee is intended
to govern together, assessing its value by the value of its individual
members is no longer appropriate, as we need to consider the externalities
the members impose on each other -- Octavian and Mark Antony individually
were the most respected men in Rome, but governing together their power
struggle tore apart the republic. Finally, assessing a committee by the min
and the max introduces additional problems in this setting. For instance, under
a multiwinner election rule a voter with min/max preferences
will strictly prefer the committee $\set{a}$ over every other committee in which $a$ is
the best element, but it is not difficult to think of a context
where larger committees are more beneficial by virtue of being large.

Throughout this paper we have replaced the requirement that a social
correspondence be onto singletons with weak viability.
Weak viability appears to be the appropriate notion in this setting;
anything weaker would imply that there exists an $a$ that does not
feature on any winning committee, and under the assumptions of strategy-proofness
such an $a$ would not affect the outcome of any election. A stronger
notion is unnecessary as weak viability is sufficient to obtain a
dictatoriality result.

The fact that we continue to evaluate committees solely on the basis of the
committee's members is the greatest weakness of this paper, but 
this weakness is unavoidable in the setting of a committee
selection rule as established in the literature: the input to our
function is a profile of linear orders over \emph{alternatives}.
We simply have no way to distinguish between a voter that ranks $a$ first, $b$
second, and feels that $a$ and $b$ would make a great pairing, and a voter
that ranks $a$ first, $b$ second, but would never want to see the two
elected together. If we cannot distinguish between them,
we must treat them the same, and thus it is impossible to handle
the externalities committee members impose on each other
within the standard framework of committee selection
rules. If we are to study more realistic preferences over
committees, we need to move to a setting where a ballot allows a voter
to express such preferences, but that brings us to the familiar problem
of how do we fit an order over $2^A$ into the space of one ballot.

Evaluating a committee by the min and the max, while problematic
in the case of general committee selection rules, turns out to be
natural and appropriate in the case of consular election rules. Voter
$i$'s favourite committee is $\set{a,b}$ such that for every $\set{c,d}$,
any member of $\set{a,b}$ is at least as good as any member of $\set{c,d}$
and one is strictly better -- in other words, Pareto dominance. If, however,
we are to extend the enquiry to larger committees, the Duggan-Schwartz model of manipulation might need further justification or to be abandoned entirely.

\bibliographystyle{plainnat}             
\bibliography{references}  

\begin{thebibliography}{21}
\providecommand{\natexlab}[1]{#1}
\providecommand{\url}[1]{\texttt{#1}}
\expandafter\ifx\csname urlstyle\endcsname\relax
  \providecommand{\doi}[1]{doi: #1}\else
  \providecommand{\doi}{doi: \begingroup \urlstyle{rm}\Url}\fi

\bibitem[Aswal et~al.(2003)Aswal, Chatterji, and Sen]{Aswal2003}
Navin Aswal, Shurojit Chatterji, and Arunava Sen.
\newblock Dictatorial domains.
\newblock \emph{Economic Theory}, 22\penalty0 (1):\penalty0 45--62, 2003.
\newblock ISSN 1432-0479.
\newblock \doi{10.1007/s00199-002-0285-8}.
\newblock URL \url{http://dx.doi.org/10.1007/s00199-002-0285-8}.

\bibitem[Barber\`{a}(1977)]{Barbera1977}
Salvador Barber\`{a}.
\newblock The manipulation of social choice mechanisms that do not leave ``too
  much" to chance.
\newblock \emph{Econometrica}, 45\penalty0 (7):\penalty0 1573--1588, 1977.
\newblock ISSN 00129682, 14680262.
\newblock URL \url{http://www.jstor.org/stable/1913950}.

\bibitem[Barber\`{a}(2010)]{Barbera2010}
Salvador Barber\`{a}.
\newblock Strategy-proof social choice.
\newblock Barcelona Economics Working Paper Series 420, Universitat
  Aut\`{o}noma de Barcelona and Barcelona GSE, 2010.
\newblock URL \url{http://digital.csic.es/handle/10261/35321}.

\bibitem[Barberà et~al.(2004)Barberà, Bossert, and Pattanaik]{BBP2004}
Salvador Barberà, Walter Bossert, and Prasanta~K. Pattanaik.
\newblock \emph{Ranking sets of objects}.
\newblock Springer, 2004.
\newblock URL
  \url{http://link.springer.com/chapter/10.1007/978-1-4020-7964-1_4}.

\bibitem[Campbell and Kelly(2002)]{Campbell2002}
Donald~E. Campbell and Jerry~S. Kelly.
\newblock A leximin characterization of strategy-proof and non-resolute social
  choice procedures.
\newblock \emph{Economic Theory}, 20\penalty0 (4):\penalty0 809--829, 2002.
\newblock ISSN 1432-0479.
\newblock \doi{10.1007/s00199-001-0239-6}.
\newblock URL \url{http://dx.doi.org/10.1007/s00199-001-0239-6}.

\bibitem[Duddy et~al.(2014)Duddy, Houy, Lang, Piggins, and Zwicker]{Duddy2014}
Conal Duddy, Nicolas Houy, J{\'e}r{\^o}me Lang, Ashley Piggins, and William~S
  Zwicker.
\newblock Social dichotomy functions, 2014.
\newblock URL
  \url{http://static.uni-graz.at/fileadmin/sowi-institute/Volkswirtschaftslehre/Wendner/Economic_Research_Seminar/GV_Zwicker_extended_abstract_Feb_22.pdf}.

\bibitem[Duggan and Schwartz(1992)]{Duggan1992}
John Duggan and Thomas Schwartz.
\newblock Strategic manipulability is inescapable: Gibbard-satterthwaite
  without resoluteness.
\newblock Working Papers 817, California Institute of Technology, Division of
  the Humanities and Social Sciences, 1992.
\newblock URL \url{http://EconPapers.repec.org/RePEc:clt:sswopa:817}.

\bibitem[Duggan and Schwartz(2000)]{Duggan2000}
John Duggan and Thomas Schwartz.
\newblock Strategic manipulability without resoluteness or shared beliefs:
  Gibbard-satterthwaite generalized.
\newblock \emph{Social Choice and Welfare}, 17\penalty0 (1):\penalty0 85--93,
  2000.
\newblock ISSN 1432-217X.
\newblock \doi{10.1007/PL00007177}.
\newblock URL \url{http://dx.doi.org/10.1007/PL00007177}.

\bibitem[Dummett and Farquharson(1961)]{DuFa1961}
Michael Dummett and Robin Farquharson.
\newblock Stability in voting.
\newblock \emph{Econometrica: Journal of The Econometric Society}, pages
  33--43, 1961.
\newblock URL \url{http://www.jstor.org/stable/1907685}.

\bibitem[Gibbard(1973)]{Gibbard1973}
Allan Gibbard.
\newblock Manipulation of voting schemes: A general result.
\newblock \emph{Econometrica}, 41\penalty0 (4):\penalty0 587--601, 1973.
\newblock ISSN 00129682, 14680262.
\newblock URL \url{http://www.jstor.org/stable/1914083}.

\bibitem[Gibbard(1977)]{Gibb1977}
Allan Gibbard.
\newblock Manipulation of {Schemes} that {Mix} {Voting} with {Chance}.
\newblock \emph{Econometrica}, 45\penalty0 (3):\penalty0 665, April 1977.
\newblock ISSN 00129682.
\newblock \doi{10.2307/1911681}.
\newblock URL \url{http://www.jstor.org/stable/1911681?origin=crossref}.

\bibitem[Kelly(1977)]{Kelly1977}
Jerry~S. Kelly.
\newblock Strategy-proofness and social choice functions without
  singlevaluedness.
\newblock \emph{Econometrica}, 45\penalty0 (2):\penalty0 439--446, 1977.
\newblock ISSN 00129682, 14680262.
\newblock URL \url{http://www.jstor.org/stable/1911220}.

\bibitem[Meir et~al.(2008)Meir, Procaccia, Rosenschein, and Zohar]{Meir2008}
Reshef Meir, Ariel~D. Procaccia, Jeffrey~S. Rosenschein, and Aviv Zohar.
\newblock Complexity of strategic behavior in multi-winner elections.
\newblock \emph{J. Artif. Int. Res.}, 33\penalty0 (1):\penalty0 149--178,
  September 2008.
\newblock ISSN 1076-9757.
\newblock URL \url{http://dl.acm.org/citation.cfm?id=1622698.1622703}.

\bibitem[Moulin et~al.(2016)Moulin, Brandt, Conitzer, Endriss, Lang, and
  Procaccia]{handbook2016}
Herv{\'e} Moulin, Felix Brandt, Vincent Conitzer, Ulle Endriss, J{\'e}r{\^o}me
  Lang, and Ariel~D Procaccia.
\newblock \emph{Handbook of Computational Social Choice}.
\newblock Cambridge University Press, 2016.

\bibitem[{\"O}zyurt and Sanver(2008)]{Ozyurt2008}
Sel{\c{c}}uk {\"O}zyurt and M.~Remzi Sanver.
\newblock Strategy-proof resolute social choice correspondences.
\newblock \emph{Social Choice and Welfare}, 30\penalty0 (1):\penalty0 89--101,
  2008.
\newblock ISSN 1432-217X.
\newblock \doi{10.1007/s00355-007-0223-6}.
\newblock URL \url{http://dx.doi.org/10.1007/s00355-007-0223-6}.

\bibitem[{\"O}zyurt and Sanver(2009)]{Ozyurt2009}
Sel{\c{c}}uk {\"O}zyurt and M.~Remzi Sanver.
\newblock A general impossibility result on strategy-proof social choice
  hyperfunctions.
\newblock \emph{Games and Economic Behavior}, 66\penalty0 (2):\penalty0 880 --
  892, 2009.
\newblock ISSN 0899-8256.
\newblock \doi{http://dx.doi.org/10.1016/j.geb.2008.09.026}.
\newblock URL
  \url{http://www.sciencedirect.com/science/article/pii/S089982560800170X}.
\newblock Special Section In Honor of David Gale.

\bibitem[Pattanaik and Dutta(1978)]{Pattanaik1978}
P.K. Pattanaik and B.~Dutta.
\newblock \emph{Strategy and group choice}.
\newblock Contributions to economic analysis. North-Holland Pub. Co., 1978.
\newblock ISBN 9780444851260.
\newblock URL \url{https://books.google.co.nz/books?id=AdDrAAAAMAAJ}.

\bibitem[Reffgen(2011)]{Reffgen2011}
Alexander Reffgen.
\newblock Generalizing the gibbard--satterthwaite theorem: partial preferences,
  the degree of manipulation, and multi-valuedness.
\newblock \emph{Social Choice and Welfare}, 37\penalty0 (1):\penalty0 39--59,
  2011.
\newblock ISSN 1432-217X.
\newblock \doi{10.1007/s00355-010-0479-0}.
\newblock URL \url{http://dx.doi.org/10.1007/s00355-010-0479-0}.

\bibitem[Satterthwaite(1975)]{Satterthwaite1975}
Mark~Allen Satterthwaite.
\newblock Strategy-proofness and arrow's conditions: Existence and
  correspondence theorems for voting procedures and social welfare functions.
\newblock \emph{Journal of Economic Theory}, 10\penalty0 (2):\penalty0 187 --
  217, 1975.
\newblock ISSN 0022-0531.
\newblock \doi{http://dx.doi.org/10.1016/0022-0531(75)90050-2}.
\newblock URL
  \url{http://www.sciencedirect.com/science/article/pii/0022053175900502}.

\bibitem[Skowron et~al.(2016)Skowron, Faliszewski, and Slinko]{Skowron2016}
Piotr Skowron, Piotr Faliszewski, and Arkadii Slinko.
\newblock Axiomatic characterization of committee scoring rules.
\newblock \emph{arXiv preprint arXiv:1604.01529}, 2016.

\bibitem[Taylor(2002)]{Taylor2002}
Alan~D. Taylor.
\newblock The manipulability of voting systems.
\newblock \emph{The American Mathematical Monthly}, 109\penalty0 (4):\penalty0
  321--337, 2002.
\newblock ISSN 00029890, 19300972.
\newblock URL \url{http://www.jstor.org/stable/2695497}.

\end{thebibliography}

\end{document}